\documentclass[conference,letterpaper]{IEEEtran}

\usepackage[utf8]{inputenc} 
\usepackage[T1]{fontenc}
\usepackage{url}
\usepackage{ifthen}
\usepackage{cite}
\usepackage[cmex10]{amsmath}
\usepackage{graphicx}
\usepackage{subcaption}
\usepackage{marginnote}
\usepackage{lipsum}
\usepackage{graphicx}
\usepackage[utf8]{inputenc}
\usepackage{tikz}
\usepackage{amsthm}
\usepackage{amsfonts}
\usepackage{cite}
\usepackage{color}
\usepackage{comment}

\newtheorem{theorem}{Theorem}[section]
\newtheorem{corollary}{Corollary}[theorem]

\newtheorem{lemma}{Lemma}[theorem]
\theoremstyle{definition}
\newtheorem{definition}{Definition}[section]

\newcommand{\x}{\mathbf{x}}
\newcommand{\y}{\mathbf{y}}
\newcommand{\A}{\mathbf{A}}
\newcommand{\w}{\mathbf{w}}
\newcommand{\X}{\mathbf{X}}
\newcommand{\ramp}{\mathbf{r}}
\newcommand{\Z}{\mathbf{Z}}

\MakeRobust\widetilde
\renewcommand{\baselinestretch}{0.97}
\begin{document}
\title{An Analysis of State Evolution for Approximate Message Passing with Side Information} 

\author{%
  \IEEEauthorblockN{Hangjin Liu}
  \IEEEauthorblockA{NC State University \\
                    Email: hliu25@ncsu.edu}
 \and
 \IEEEauthorblockN{Cynthia Rush}
  \IEEEauthorblockA{Columbia University\\
                    Email: cynthia.rush@columbia.edu}
 \and
 \IEEEauthorblockN{Dror Baron}
  \IEEEauthorblockA{NC State University\\
                    Email: barondror@ncsu.edu}
}

\maketitle
\begin{abstract}
A common goal in many research areas is to reconstruct an unknown signal $\x$ from noisy linear measurements.  Approximate message passing (AMP) is a class of low-complexity algorithms for efficiently solving such high-dimensional regression tasks.  Often, it is the case that side information (SI) is available during reconstruction.  For this reason a novel algorithmic framework that incorporates SI into AMP, referred to as approximate message passing with side information (AMP-SI), has been recently introduced.   An attractive feature of AMP is that when the elements of the signal are exchangeable, the entries of the measurement matrix are independent and identically distributed (i.i.d.) Gaussian, and the denoiser applies the same non-linearity at each entry, the performance of AMP can be predicted accurately by a scalar iteration referred to as state evolution (SE).  However, the AMP-SI framework uses different entry-wise scalar denoisers, based on the entry-wise level of the SI, and therefore is not supported by the standard AMP theory. In this work, we provide rigorous performance guarantees for AMP-SI when the input signal and SI are drawn i.i.d.\ according to some joint distribution subject to finite moment constraints. Moreover, we provide numerical examples to support the theory which demonstrate empirically that the SE can predict the AMP-SI mean square error accurately.
\end{abstract}

\section{Introduction}
High-dimensional linear regression is a well-studied model being used in many applications including compressed sensing\cite{DMM2009}, 
imaging\cite{Arguello2011}, and machine learning and statistics\cite{Hastie2001}. 
The unknown signal $\x \in\mathbb{R}^n$ is viewed through the linear model:
\begin{equation}\label{eq:1-1}
\y=\A \x+ \w,
\end{equation}
where $\mathbf{y}\in\mathbb{R}^m$ are the measurements, $\A\in\mathbb {R}^{m\times n}$ is a known measurement matrix, and  $\w\in \mathbb R^m$ is measurement noise. The goal is to estimate the unknown signal $\x$ having knowledge only of the noisy measurements $\y$ and the measurement matrix $\A$.  When the problem is under-determined (i.e., $m<n$), in order for reconstruction to be successful, it is necessary to exploit structural or probabilistic characteristics of the input signal $\x$.  Often a prior distribution on the input signal $\x$ is assumed, and in this case approximate message passing (AMP) algorithms\cite{DMM2009} can be used for the reconstruction task.

AMP~\cite{DMM2009, RanganGAMP2010} is a class of low-complexity algorithms for efficiently solving high-dimensional regression tasks (\ref{eq:1-1}). AMP works by iteratively generating estimates of the unknown input vector,  $\x$, using a possibly non-linear denoiser function tailored to any prior knowledge about $\x$. 
One favorable feature of AMP is that under some technical conditions on the measurement matrix $\A$ and $\x$, the observations at each iteration of the algorithm are almost surely equal in distribution to $\x$ plus independent and identically distributed (i.i.d.) Gaussian noise in the large system limit.

\textbf{AMP with Side Information (AMP-SI):} In information theory~\cite{Cover06}, when different communication systems share side information (SI), overall communication can become more efficient. Recently~\cite{Dror2017, Ma2018}, a novel algorithmic framework, referred to as AMP-SI, has been introduced for incorporating SI into AMP for high-dimensional regression tasks (\ref{eq:1-1}).  AMP-SI has been empirically demonstrated to have good reconstruction quality and is easy to use. For example, we have proposed to use AMP-SI for channel estimation in emerging millimeter wave communication systems~\cite{saleh1987statistical}, where the time dynamics of the channel structure allow previous channel estimates to be used as SI when estimating the current channel structure~\cite{Ma2018}.

We model the observed SI, denoted by ${{\widetilde{\x}}\in\mathbb{R}^n}$,  as depending statistically on the unknown signal $\x$ through some joint probability density function (pdf), $f({\X},\widetilde{{\X}})$. AMP-SI uses a conditional denoiser, $g_t:\mathbb{R}^{2n}\rightarrow \mathbb{R}^n$, to incorporate SI,
\begin{equation}\label{eq:eta_2}
g_t(\mathbf{a}, \mathbf{b})=\mathbb{E}[\X | \X+\lambda_t\mathcal{N}(0,
\mathbb{I}_n)=\mathbf{a}, \widetilde \X= \mathbf{b}].
\end{equation}

The AMP-SI algorithm iteratively updates estimates of the input signal $\x$: let $\x^0  = \mathbf{0}$, the all-zeros vector, then
\begin{align}
\ramp^t &=\y-\A\x^t+\frac{\ramp^{t-1}}{\delta} [ \text{div} \, g_{t-1}(\x^{t-1}+\A^T\ramp^{t-1}, \widetilde \x)], \label{eq:1-5} \\
 \x^{t+1}&={g_{t}}(\x^{t}+\A^T\ramp^{t}, \widetilde \x), \label{eq:1-6}
\end{align}
where $\x^t\in \mathbb{R}^n$ is the estimate of $\x$ at iteration $t$ and $\delta=\frac{m}{n}$ is the measurement rate. For a differential function $g: \mathbb{R}^{2n} \rightarrow \mathbb{R}^n$ we use $\text{div} g(\mathbf{a}, \mathbf{b}) = \sum_{i=1}^n \frac{\partial g_i}{\partial a_i}(\mathbf{a}, \mathbf{b})$. Using the denoiser in \eqref{eq:eta_2}, the AMP-SI algorithm~\eqref{eq:1-5}-\eqref{eq:1-6} provides the minimum mean squared error (MMSE) estimate of the signal when SI  $\widetilde{\x}$ is available~\cite{Dror2017}.

\textbf{State Evolution (SE):} It has been proven that the performance of AMP, as measured, for example, by the normalized squared $\ell_2$-error $\frac{1}{n}||\x^t-\x||_2^2$ between the estimate  $\x^t$ and true signal $\x$, can be accurately  predicted by a scalar recursion referred as SE\cite{Bayati2011,Rush_Finite18} when the measurement matrix $\A$ is i.i.d.\ Gaussian under various assumptions on the elements of the signal. The SE equation for AMP-SI is as follows. Assume the entries of the noise $\w$ are i.i.d.\ $\sim f(W)$ with $\sigma_w^2 = \mathbb{E}[W^2]$, and let $\lambda_0 = \sigma_w^2 + \mathbb{E}[||\X||^2]/n\delta$. Then for $t \geq 0$,
\begin{equation}
\lambda_t^2 = \sigma_w^2 + \frac{1}{\delta n}\mathbb{E}\left[||g_{t-1}(\X + \lambda_{t-1}\Z, \widetilde{\X}) - \X||^2\right],
\label{eq:SE2}
\end{equation}
where $(\X, \widetilde{\X}) \sim f(\X, \widetilde{\X})$ are independent of $\Z\sim \mathcal{N}(0,\mathbb{I}_n)$, where we use $\mathcal{N}(\mu, \sigma^2)$ to denote a Gaussian distribution with mean $\mu$ and variance $\sigma^2$.  

Considering AMP-SI~\eqref{eq:1-5}-\eqref{eq:1-6}, however, we cannot directly apply the existing AMP theoretical results~\cite{Bayati2011,Rush_Finite18}, as the conditional denoiser~\eqref{eq:eta_2} depends on the index $i$ through the SI, meaning that different scalar denoisers  will be used at different indices within the AMP-SI iterations.  Recent results~\cite{Berthier2017}, however, extend the asymptotic SE analysis to a larger class of possible denoisers, allowing, for example, each element of the input to use a different non-linear denoiser as is the case in AMP-SI.  We employ these results to rigorously relate the SE presented in \eqref{eq:SE2} to the AMP-SI algorithm in~\eqref{eq:1-5}-\eqref{eq:1-6}.

{\bf Related Work:}
While integrating SI into reconstruction algorithms is not new, AMP-SI introduces a unified framework within AMP supporting arbitrary signal and SI dependencies.
Prior work using SI has been either heuristic, limited 
to specific applications, or outside the AMP framework. 

For example, Wang and Liang~\cite{wang2015approximate} integrate SI into AMP for a specific signal prior density, but the method is difficult to apply to other signal models.  Ziniel and Schniter~\cite{DCSAMP} develop an AMP-based reconstruction algorithm for a
time-varying signal model based on Markov processes for the support and amplitude.  This signal model is easily incorporated into the AMP-SI framework as discussed in the analysis of the birth-death-drift model of~\cite{Dror2017, Ma2018}.  Manoel et al. implement an AMP-based algorithm in which
the input signal is repeatedly reconstructed in a streaming fashion,
and information from past reconstruction attempts is aggregated into a prior,
thus improving ongoing reconstruction results~\cite{manoel2017streaming}.
This reconstruction scheme resembles that of AMP-SI, in particular when the Bernoulli-Gaussian model is used (see Section~\ref{subsec:examples}).

{\bf Contribution and Outline:} Ma et al.\ use numerical experiments to show that SE~\eqref{eq:SE2} accurately tracks the performance of AMP-SI~\eqref{eq:1-5}-\eqref{eq:1-6}~\cite{Ma2018}, as was shown rigorously for standard AMP. Ma et al.\ conjecture that rigorous theoretical guarantees can be given for AMP-SI as well~\cite{Ma2018}. In this work, we analyze AMP-SI performance when the input signal and SI are drawn i.i.d.\ according to a general pdf $f(\X, \widetilde{\X})$ obeying some finite moment conditions, the AMP-SI denoiser~\eqref{eq:eta_2} is Lipschitz, and the measurement matrix $\A$ is i.i.d.\ Gaussian.

In Section~\ref{main_result}, we give the main results, examples for various signal and SI models, and numerical experiments comparing the empirical performance of AMP-SI and the SE predictions. The proof of our main theorem is provided in Section~\ref{main_proof}.

\section{Main Results}\label{main_result}
\subsection {Main Theorem}
Our main result provides AMP-SI performance guarantees when considering \textit{pseudo-Lipschitz} loss functions, which we define in the following.
\begin{definition}
\label{def:PLfunc}
\textbf{\emph{Pseudo-Lipschitz functions}}~\cite{Berthier2017}: For $k\in \mathbb{N}_{>0}$ and any $n\in \mathbb{N}_{>0}$, a function $\phi : \mathbb{R}^n\to
\mathbb{R}$ is \emph{pseudo-Lipschitz of order $k$}, or \emph{PL(k)}, if there exists a constant $L$, referred to as the pseudo-Lipschitz constant of $\phi$, such that for $\x, \y \in \mathbb{R}^n$
\begin{equation*} 
\left|\phi(\x)-\phi(\y)\right|\leq L\Big(1+ \Big(\frac{||\x||}{\sqrt{n}}\Big)^{k-1}+ \Big(\frac{||\y||}{\sqrt{n}}\Big)^{k-1}\Big) \frac{||\x-\y||}{\sqrt{n}}.
\end{equation*}
For $k = 1$, this definition coincides with the standard definition of a Lipschitz function. %

A sequence (in $n$) of PL(k) functions $\{\phi_n\}_{n\in \mathbb{N}_{>0}}$
is called \emph{uniformly pseudo-Lipschitz}
of order $k$ , or \emph{uniformly PL(k)}, if, denoting by $L_n$ the pseudo-Lipschitz constant of $\phi_n$, we have $L_n < \infty$ for each $n$ and $\lim\sup_{n\to\infty}L_n < \infty$.
\end{definition}
Throughout the work, $||\cdot||$ denotes the  Euclidean norm, and $\overset{p}{=}$ denotes convergence in probability.
In the case of $(\X, \widetilde{\X})$ sampled i.i.d.\ $f(X, \widetilde{X})$ the AMP-SI denoiser (originally defined in \eqref{eq:eta_2}) is separable: define $\eta_t: \mathbb{R}^2 \rightarrow \mathbb{R}$, as
\begin{equation}\label{eq:eta_2_iid}
\eta_t(a, b)=\mathbb{E}[X | X+\lambda_t\mathcal{N}(0,
1)=a, \widetilde X= b],
\end{equation}
and the AMP-SI algorithm in \eqref{eq:1-5}-\eqref{eq:1-6} simplifies to
\begin{align}
&\ramp^t =\y-\A\x^t+\frac{\ramp^{t-1}}{\delta} \sum_{i=1}^n \eta'_{t-1}([\x^{t-1}+\A^T\ramp^{t-1}]_i, \widetilde x_i), \label{eq:1-5_iid} \\
& x_i^{t+1} ={\eta_{t}}([\x^{t}+\A^T\ramp^{t}]_i, \widetilde x_i), \quad \text{ for } i = 1, 2, \ldots, n, \label{eq:1-6_iid}
\end{align}
where the derivative $\eta_t'(s, \cdot) = \frac{\partial}{\partial s} \eta_t(s, \cdot)$. For the denoiser in \eqref{eq:eta_2_iid}, the SE is as follows: let $\lambda_0 = \sigma_w^2 + \mathbb{E}[X^2]/\delta$ and for $t \geq 0$,
\begin{equation}
\lambda_t^2 = \sigma_w^2 + \frac{1}{\delta}\mathbb{E}\left[(\eta_{t-1}(X + \lambda_{t-1}Z, \widetilde{X}) - X)^2\right],
\label{eq:SE2_iid}
\end{equation}
where $(X, \widetilde{X}) \sim f(X, \widetilde{X})$ are independent of $Z\sim \mathcal{N}(0,1)$.
\begin{theorem}
\label{thm:SE}
For any PL(2) functions $\phi: \mathbb{R}^{2} \rightarrow \mathbb{R}$ and $\psi: \mathbb{R}^{3} \rightarrow \mathbb{R}$, define sequences of  functions $\phi_m: \mathbb{R}^{2m} \rightarrow \mathbb{R}$ and $\psi_n: \mathbb{R}^{3n} \rightarrow \mathbb{R}$ as follows: for vectors $\mathbf{a}, \mathbf{b} \in \mathbb{R}^m$ and $\x, \y, \widetilde{\x} \in \mathbb{R}^n$,
\begin{equation}
\begin{split}
 \phi_m(\mathbf{a}, \mathbf{b}) &:= \frac{1}{m} \sum_{i=1}^m \phi(a_i, b_i) \\
  \psi_n(\x, \y,\widetilde{\x}) &:= \frac{1}{n} \sum_{i=1}^n \psi(x_i, y_i,\widetilde{x}_i).
 \label{eq:sum_funcs}
\end{split}
\end{equation}
Then the functions in \eqref{eq:sum_funcs} are uniformly PL(2). Next, assume the following:
\begin{itemize}
\item[\textbf{(A1)}] The measurement matrix $\A$ has i.i.d.\ Gaussian entries with mean $0$ and variance $1/m$.
\item[\textbf{(A2)}] The noise $\w$ is i.i.d.\ $\sim f(W)$ with finite $\mathbb{E}[|W|^{\max\{k,2\}}]$.
\item[\textbf{(A3)}] The signal and SI $(\x, \widetilde{\x})$ are sampled i.i.d.\ from $f(X, \widetilde{X})$ with finite $\mathbb{E}[|X|^{2}]$, finite $ \mathbb{E}[|\widetilde{X}|^{2}]$, and finite $\mathbb{E}[|X \widetilde{X}|]$.
\item[\textbf{(A4)}] For $t\geq 0$, the denoisers $\eta_t(\cdot, \cdot)$ defined in \eqref{eq:eta_2_iid} are Lipschitz continuous: for scalars $a_1, a_2, b_1, b_2$, and constant $L > 0$,
$|\eta_t(a_1, b_1) - \eta_t(a_2, b_2)| \leq L||(a_1, b_1) - (a_2, b_2)||$.
\end{itemize}
Then, we have the following asymptotic results for the functions defined in \eqref{eq:sum_funcs},
\begin{equation}\label{eq:main}
\begin{aligned}
&\lim_m\phi_m(\ramp^t, {\w}) \overset{p}{=} \lim_m  \mathbb{E}[\phi_m(\mathbf{W} + \sqrt{\lambda_t^2 - \sigma_w^2} \, {\Z_1}, {\mathbf{W}})]),\\
&\lim_n \psi_n(\x^t + {\A}^T {\ramp}^t,\x, \widetilde{\x})  \overset{p}{=} \lim_n \mathbb{E}[\psi_n({\mathbf{X}} + \lambda_t {\Z_2}, {\mathbf{X}},\widetilde{\mathbf{X}})] ,
\end{aligned}
 \end{equation}
 where $\Z_1\sim\mathcal{N}(0,\mathbb{I}_m)$, $\Z_2\sim \mathcal{N}(0,\mathbb{I}_n)$, independent of $\mathbf{W}  \sim i.i.d.\ f(W)$ and $(\mathbf{X}, \widetilde{\mathbf{X}})  \sim i.i.d.\ f(X, \widetilde{X})$.  $\x^t$ and $\ramp^t$ are defined in the AMP-SI recursion~\eqref{eq:1-5_iid}-\eqref{eq:1-6_iid}, and $\lambda_t$ in the SE~\eqref{eq:SE2_iid}.
\end{theorem}

Section~\ref{main_proof} contains the proof of Theorem~\ref{thm:SE}. The proof follows from Berthier et al. [11, Theorem 14] and the strong law of large numbers.  The main details involve showing that assumptions $\textbf{(A1)} - \textbf{(A4)}$ allow us to apply~\cite[Theorem 14]{Berthier2017}

As a concrete example of how Theorem~\ref{thm:SE} provides performance guarantees for AMP-SI, let us consider a few interesting pseudo-Lipschitz loss functions.

\begin{corollary}
\label{cor:SE}
Under assumptions $\textbf{(A1)} - \textbf{(A4)}$, letting $\psi^1: \mathbb{R}^{3} \rightarrow \mathbb{R}$ be $\psi^1(x, y, z) = (x-y)^2$, then by Theorem~\ref{thm:SE}, 
\[\lim_{n \rightarrow \infty}  \frac{1}{n} ||\x^t + \A^T \ramp^t- \x||^2 \overset{p}{=}  \lambda_t^2,\]
where $\lambda_t^2$ is defined in \eqref{eq:SE2}.
Similarly if $\psi^2: \mathbb{R}^{3} \rightarrow \mathbb{R}$ is defined as $\psi^2(x, y, z) = (\eta_t(x, z)-y)^2$, then by Theorem~\ref{thm:SE} 
\[\lim_{n \rightarrow \infty}  \frac{1}{n} ||\x^{t+1}- \x||^2 \overset{p}{=}  \delta(\lambda_{t+1}^2 - \sigma_w^2).\]
When $\eta_t$ is Lipschitz, it is straightforward to show that $\psi^1$ and $\psi^2$ are both PL(2), and thus Theorem~\ref{thm:SE} can be applied.
\end{corollary}

\subsection{Examples} \label{subsec:examples}
Next, we consider a few signal and SI models to show how one can derive the denoiser in \eqref{eq:eta_2}, use this to construct the AMP-SI algorithm and the SE, and apply Theorem \ref{thm:SE}. Before we get to the examples we state a lemma that allows us know about how functions with bounded derivative are Lipschitz.
\begin{lemma}
\label{lem:Lipschitz}
A function $\phi: \mathbb{R}^2 \rightarrow \mathbb{R}$ having bounded derivatives, $0 < \mathsf{D}_1,  \mathsf{D}_2 < \infty,$
\[\Big \lvert \frac{\partial}{\partial x} \phi(x, y) \Big \lvert \leq \mathsf{D}_1 \qquad \text{and } \qquad \Big \lvert \frac{\partial}{\partial y} \phi(x, y)\Big \lvert \leq \mathsf{D}_2\]
 is Lipschitz continuous with Lipschitz constant $\sqrt{\mathsf{D}_1^2 + \mathsf{D}_2^2}$.
\end{lemma}
\begin{proof}
The result follows using the Triangle Inequality and Cauchy-Schwarz,
\begin{equation}
\begin{split}
&\lvert \phi(x_1, y_1) - \phi(x_2, y_2) \lvert
\\ &= \lvert \phi(x_1, y_1) - \phi(x_1, y_2) + \phi(x_1, y_2) - \phi(x_2, y_2) \lvert  \\
&\leq \lvert \phi(x_1, y_1) - \phi(x_1, y_2) \lvert + \lvert \phi(x_1, y_2) - \phi(x_2, y_2) \lvert  \\
&\leq  \mathsf{D}_2 \lvert y_1- y_2 \lvert +  \mathsf{D}_1 \lvert x_1 - x_2 \lvert \\
&\leq  \sqrt{\mathsf{D}_2^2+  \mathsf{D}_1^2} \sqrt{ ( y_1- y_2)^2  + (x_1 - x_2)^2} \\
&=  \sqrt{\mathsf{D}_2^2+  \mathsf{D}_1^2}||{(x_1, y_1) - (x_2, y_2)||}.
\end{split}
\end{equation}
\end{proof}

\subsubsection{Gaussian-Gaussian Signal and SI}
In this model, referred to as the GG model henceforth, the signal has i.i.d. Gaussian entries with zero mean and finite variance and we have access to SI in the form of the signal with additive white Gaussian noise (AWGN). The signal, $X$, and SI, $\widetilde X$, are related by
\begin{equation}\label{eq:SI form}
    \widetilde X=X+ \mathcal{N}(0, \sigma^2\mathbb{I}).
\end{equation}
In this case, the AMP-SI denoiser~\eqref{eq:eta_2} equals~\cite{Ma2018}
\begin{equation}
\begin{aligned}
\eta_{t}(a,b)&=\mathbb{E}\left[X\middle| X+\lambda_{t}Z=a, \widetilde X=b\right]\\& = \frac{\sigma_x^2 \sigma^2 a + \sigma_x^2 \lambda_{t}^2 b}{\sigma_x^2 (\sigma^2 + \lambda_{t}^2) + \sigma^2 \lambda_{t}^2}.
\label{eq:GGdenoiser}
\end{aligned}
\end{equation} 
Then the SE~\eqref{eq:SE2} can be computed as
\begin{equation}
\begin{split}
\lambda_t^2 = \sigma_w^2 + \frac{1}{\delta}\left[ \frac{ \sigma_x^2 \sigma^2 \lambda_{t-1}^2 }{\sigma_x^2 (\sigma^2 + \lambda_{t-1}^2) + \sigma^2 \lambda_{t-1}^2} \right] .
\end{split}
\end{equation}
We note that  as a result of Lemma \ref{lem:Lipschitz} because
\[\Big \lvert \frac{\partial}{\partial a} \eta_{t}(a,b) \Big \lvert = \Big \lvert \frac{\sigma_x^2 \sigma^2}{\sigma_x^2 (\sigma^2 + \lambda_{t}^2) + \sigma^2 \lambda_{t}^2} \Big \lvert \leq 1,\]
and
\[\Big \lvert \frac{\partial}{\partial b} \eta_{t}(a,b) \Big \lvert = \Big \lvert \frac{ \sigma_x^2 \lambda_{t}^2}{\sigma_x^2 (\sigma^2 + \lambda_{t}^2) + \sigma^2 \lambda_{t}^2} \Big \lvert \leq 1,\]
and therefore the assumptions $\textbf{(A1)} - \textbf{(A4)}$ are satisfied in the GG case and we can apply Thoerem~\ref{thm:SE}.
\subsubsection{Bernoulli-Gaussian Signal and SI}
The Bernoulli-Gaussian (BG) model reflects scenario in which one wishes to recover a sparse signal and has access to SI in the form of the signal with AWGN as in \eqref{eq:SI form}. In this model, each entry of the signal is independently generated according to $x_i\sim \epsilon\mathcal{N}(0,1)+(1-\epsilon)\delta_0$, where $\delta_0$ is the Dirac delta function at $0$. In words, the entries of the signal independently take the value $0$ with probability $1-\epsilon$ and are $\mathcal{N}(0,1)$ with probability $\epsilon$. 
In this case, the AMP-SI denoiser~\eqref{eq:eta_2} equals~\cite{Ma2018}
\begin{equation}
\begin{aligned}
\eta_{t}(a,b)&=\mathbb{E}\left[X\middle| X+\lambda_{t}Z=a, \widetilde X=b\right]\\&=\Pr \left (X\neq0|a, b \right)\mathbb{E}\left[X\middle|a, b, X\neq 0\right] \\&=\Pr \left (X\neq0|a, b \right) \frac{ \sigma^2 a + \lambda_{t}^2 b}{ \sigma^2 + \lambda_{t}^2 + \sigma^2 \lambda_{t}^2},
\label{eq:BGdenoiser}
\end{aligned}
\end{equation} 
where, letting $\rho_{\tau^2}(x)$ be the zero-mean Gaussian density with variance $\tau^2$ evaluated at $x$, and defining $\nu_{t} := \sigma^2 \lambda_{t}^2 (\sigma^2+\lambda_{t}^2 + \sigma^2 \lambda_{t}^2)$,
\begin{equation}
\begin{split}
\Pr(X\neq0|a, b) &=\left( 1+ T_{a,b}\right)^{-1},
\label{eq:BGprobability}
\end{split}
\end{equation}
where we denote
\begin{equation}
\begin{split}
T_{a,b}&:=\frac{(1 - \epsilon) \rho_{\lambda_{t}^2}(a) \rho_{\sigma^2}(b) }{ \epsilon\rho_{1 + \sigma^2}(b)   \rho_{\frac{\sigma^2}{1 + \sigma^2} + \lambda_{t}^2}\Big(\frac{b}{1 + \sigma^2}  - a\Big)} \\
&= \Big(\frac{1 - \epsilon}{\epsilon}\Big) \sqrt{\frac{\sigma^2+\lambda_{t}^2 + \sigma^2 \lambda_{t}^2}{ \lambda_{t}^2 \sigma^2}}  \\
&\qquad\exp\Big\{\frac{-( \sigma^2 a + \lambda_{t-1}^2 b)^2}{2 \sigma^2 \lambda_{t-1}^2 (\sigma^2+\lambda_{t}^2 + \sigma^2 \lambda_{t}^2)}\Big\}\\
&= \Big(\frac{1 - \epsilon}{\epsilon}\Big) \frac{\nu_{t}  \sqrt{2 \pi }}{ \lambda_{t}^2 \sigma^2}  \rho_{\nu_{t} }( \sigma^2 a + \lambda_{t}^2 b),
\label{eq:Tab_def}
\end{split}
\end{equation}
Then the SE~\eqref{eq:SE2} can be computed as
\begin{equation}
\begin{split}
\lambda_t^2 =\sigma_w^2 + \frac{1}{\delta}\Big(\frac{T_{a,b}}{1 + T_{a,b}}\Big)^2 \left[ \frac{ (\sigma^2+\lambda_{t-1}^2)+ \sigma^2 \lambda_{t-1}^2 }{\sigma^2 + \lambda_{t-1}^2 + \sigma^2 \lambda_{t-1}^2} \right].
\label{eq:BG_SE}
\end{split}
\end{equation}
We again use Lemma \ref{lem:Lipschitz} to show that the denoiser defined in \eqref{eq:BGdenoiser} and \eqref{eq:BGprobability} is Lipschitz continuous so that the assumptions $\textbf{(A1)} - \textbf{(A4)}$ are satisfied in the BG case and we can apply Thoerem~\ref{thm:SE}. We study the partial derivatives. 
Denote
\begin{equation}
f_{a,b} :=\frac{ \sigma^2 a + \lambda_{t}^2 b}{ \sigma^2 + \lambda_{t}^2 + \sigma^2 \lambda_{t}^2}.
\label{eq:fab_def}
\end{equation}
Combining \eqref{eq:BGprobability} and \eqref{eq:Tab_def} and \eqref{eq:fab_def}, 
\begin{equation*}
\begin{aligned}
\eta_{t}(a,b) &= (1 + T_{a,b})^{-1} f_{a,b}.
\end{aligned}
\end{equation*} 
Then,
\begin{equation}
\begin{split}
\label{eq:partial_BG_1_V1}
&\Big \lvert \frac{\partial \eta_{t}(a,b) }{\partial a}  \Big \lvert  \\
&= \Big \lvert\frac{1}{1 + T_{a,b}} \Big[ \frac{\partial f_{a,b}}{\partial a} \Big] -\frac{1}{(1+T_{a,b})^{2}} \Big[\frac{\partial T_{a,b}}{\partial a}\Big]  f_{a,b}\Big \lvert \\
&\leq \frac{ (1 + 2T_{a,b} )}{(1 + T_{a,b})^{2}}\Big \lvert  \frac{\partial f_{a,b}}{\partial a} \Big\lvert +  \frac{1}{(1 + T_{a,b})^{2}}  \Big \lvert   \frac{\partial (  T_{a,b} f_{a,b})}{\partial a} \Big \lvert.
\end{split}
\end{equation}
Now we show upperbounds for the two terms of \eqref{eq:partial_BG_1_V1} separately.  For the first term, we see that $\frac{\partial f_{a,b}}{\partial a} \leq 1$, so 
\[\frac{ (1 + 2T_{a,b} )}{(1 + T_{a,b})^{2}}\Big \lvert  \frac{\partial f_{a,b}}{\partial a} \Big\lvert \leq 1.\]
Now we consider the second term of 

Consider the second term of \eqref{eq:partial_BG_1_V1}. First we note that
\[ \frac{1}{(1 + T_{a,b})^{2}} \Big \lvert  T_{a,b} \Big[ \frac{\partial}{\partial a} f_{a,b}\Big] + \Big[\frac{\partial}{\partial a} T_{a,b}\Big]f_{a,b}  \Big \lvert \leq  \Big \lvert \frac{\partial}{\partial a}  \Big[ T_{a,b} f_{a,b}\Big]  \Big \lvert.\]
Then from \eqref{eq:Tab_def} and \eqref{eq:fab_def},
\begin{equation*}
\begin{split}
T_{a,b} f_{a,b} 
&= \Big(\frac{1 - \epsilon}{\epsilon}\Big) \sqrt{2 \pi}(\sigma^2 a + \lambda_{t}^2 b) \rho_{ \nu_{t} }(\sigma^2 a + \lambda_{t}^2 b),
\end{split}
\end{equation*}
then using that $\frac{\partial}{\partial x} \rho_{\tau^2}(x) = - \frac{x}{\tau^2}  \rho_{\tau^2}(x)$, we have
\begin{equation}
\begin{split}
\label{eq:partial_BG_1}
\Big \lvert \frac{\partial}{\partial a}  \Big[ T_{a,b} f_{a,b}\Big]   \Big \lvert  &= \Big(\frac{1 - \epsilon}{\epsilon}\Big) \sqrt{2 \pi}\Big \lvert \sigma^2  \rho_{ \nu_{t} }(\sigma^2 a + \lambda_{t}^2 b) \\& -\frac{\sigma^2(\sigma^2 a + \lambda_{t}^2 b)^2}{\nu_{t}}  \rho_{ \nu_{t} }(\sigma^2 a + \lambda_{t}^2 b)\Big \lvert \\
&= \Big(\frac{1 - \epsilon}{\epsilon}\Big) \frac{\sqrt{2 \pi} \sigma^2}{\nu_{t}}  \rho_{ \nu_{t} }(\sigma^2 a + \lambda_{t}^2 b) \\& \qquad\Big \lvert \nu_{t} - (\sigma^2 a + \lambda_{t}^2 b)^2 \Big \lvert.
\end{split}
\end{equation}
To upper bound the above, we use $\exp\{-x\} \leq \frac{1}{1+x}$ when $x \geq 0$, and so
\[\rho_{\tau^2}(x) = \frac{1}{\sqrt{2\pi \tau^2 }} \exp\Big\{\frac{-x^2}{2 \tau^2}\Big\} \leq \sqrt{\frac{2}{\pi}} \Big(\frac{\tau}{2 \tau^2 + x^2}\Big).\]
Using this in \eqref{eq:partial_BG_1}, we find
\begin{equation}\label{eq:lastpartial_a_bound}
\begin{split}
\Big \lvert \frac{\partial}{\partial a}  \Big[ T_{a,b} f_{a,b}\Big]  \Big \lvert  &\leq \frac{2 \sigma^2}{\sqrt{ \nu_{t}}} \Big(\frac{1 - \epsilon}{\epsilon}\Big)\frac{  \lvert \nu_{t} - (\sigma^2 a + \lambda_{t}^2 b)^2 \lvert}{2  \nu_{t} + (\sigma^2 a + \lambda_{t}^2 b) ^2}  \\&\leq \frac{2 \sigma^2}{\sqrt{ \nu_{t}}} \Big(\frac{1 - \epsilon}{\epsilon}\Big) \leq \frac{2(1 - \epsilon)}{ \sigma_w  \epsilon},
\end{split}
\end{equation}
where in the final inequality we use $ \lambda_{t} \geq \sigma_w$ by \eqref{eq:BG_SE}, and 
\begin{equation}
\begin{aligned}
\frac{ \sigma^2}{\sqrt{ \nu_{t}}}& = \frac{ \sigma}{  \lambda_{t} \sqrt{ \sigma^2+\lambda_{t}^2 + \sigma^2 \lambda_{t}^2}} \\&= \frac{ 1}{  \lambda_{t} \sqrt{ 1+ \frac{\lambda_{t}^2}{\sigma^2} + \lambda_{t}^2}}   \leq \frac{ 1}{  \lambda_{t}}.
\end{aligned}
\end{equation}
Using the above in \eqref{eq:partial_BG_1_V1}, we have 
\begin{equation*} 
\Big\lvert \frac{\partial}{\partial a} \eta_{t}(a,b) \Big\lvert \leq  1+  \frac{2(1 - \epsilon)}{\sigma_w \epsilon} .
\end{equation*}
As in \eqref{eq:partial_BG_1_V1} we can show
\begin{equation*}
\begin{split}
\Big \lvert \frac{\partial}{\partial b} \eta_{t}(a,b) \Big \lvert  &\leq \frac{ (1 + 2T_{a,b} )}{(1 + T_{a,b})^{2}}\Big \lvert \Big[ \frac{\partial}{\partial b} f_{a,b}\Big] \Big\lvert \\&+  \frac{1}{(1 + T_{a,b})^{2}} \Big \lvert  T_{a,b} \Big[ \frac{\partial}{\partial b} f_{a,b}\Big] + \Big[\frac{\partial}{\partial b} T_{a,b}\Big]f_{a,b}  \Big \lvert,
\end{split}
\end{equation*}
Then,
\begin{equation*}
\begin{split}
 \frac{ (1 + 2T_{a,b} )}{(1 + T_{a,b})^{2}}\Big \lvert \frac{\partial}{\partial b} f_{a,b} \Big\lvert  &\leq  1,
\end{split}
\end{equation*}
and a bound as in \eqref{eq:partial_BG_1} - \eqref{eq:lastpartial_a_bound} gives
\begin{equation*}
\begin{split}
 \frac{1}{(1 + T_{a,b})^{2}} &\Big \lvert  T_{a,b} \Big[ \frac{\partial}{\partial b} f_{a,b}\Big] + \Big[\frac{\partial}{\partial b} T_{a,b}\Big]f_{a,b}  \Big \lvert \leq  \Big \lvert  \frac{\partial}{\partial b}  \Big[ T_{a,b} f_{a,b}\Big]   \Big \lvert \\
 &\leq  \frac{2  \lambda_{t}^2}{\sqrt{\nu_{t} }} \Big(\frac{1 - \epsilon}{\epsilon}\Big) \frac{ \lvert \nu_{t} - (\sigma^2 a + \lambda_{t}^2 b)^2 \lvert}{2 \nu_{t}  + (\sigma^2 a + \lambda_{t}^2 b)^2} \\
 &\leq \frac{2  \lambda_{t}^2}{\sqrt{\nu_{t} }} \Big(\frac{1 - \epsilon}{\epsilon}\Big) \leq \frac{2(1 - \epsilon)}{ \sigma  \epsilon} .
\end{split}
\end{equation*}
\subsection{Numerical Examples} 
Finally, we provide numerical results to compare the empirical mean square error (MSE) performance of AMP-SI and the performance predicted by SE. Fig.\ \ref{fig:GG} shows the MSE achieved by AMP-SI in the GG scenario and the SE prediction of its performance. In this example, the signal variance $\sigma_x^2=1$, the measurement noise variance $\sigma_{w}^2=0.01$, the variance of AWGN in SI $\sigma^2=0.04$. We averaged over 10 trials of a GG recovery problem for empirical results of AMP-SI. The comparison in Fig.~\ref{fig:GG}(a), Fig.~\ref{fig:GG}(b) and Fig.~\ref{fig:GG}(c) given by three different signal length. For smaller $n$ there is some gap between the empirical MSE and the SE prediction, as shown in Fig.\ \ref{fig:GG} for $n=100$, but the gap shrinks as $n$ is increased. The results show the empirical MSE tracks the SE prediction nicely.

Fig.\ \ref{fig:BG} shows the MSE achieved by AMP-SI in the BG scenario, and the SE prediction of its performance. 
We again averaged over 10 trials of a BG recovery problem for empirical results of AMP-SI. The signal length $n=10000$, $m=3000$, the measurement noise variance $\sigma_{w}^2=0.01$, and $\epsilon=0.2$, where $20\%$ of the entries in the signal are nonzero. We vary the variance of AWGN in SI from $\sigma^2=0.04$, $\sigma^2=0.25$, and $\sigma^2=1$. The results show that SE can predict the MSE achieved by AMP-SI at every iteration.
\begin{figure}
\centering
\subcaptionbox{n=100}{\includegraphics[width=0.4\textwidth]{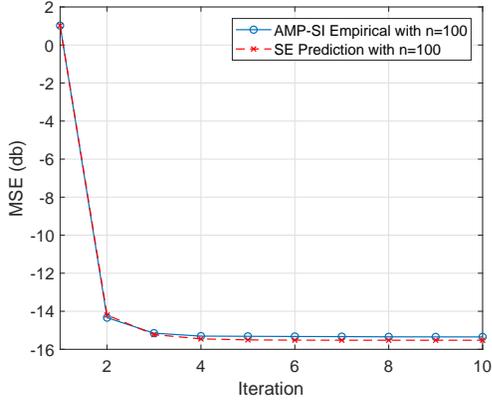}\label{GG_n=100}}%
\hfill
\subcaptionbox{n=1000}{\includegraphics[width=0.4\textwidth]{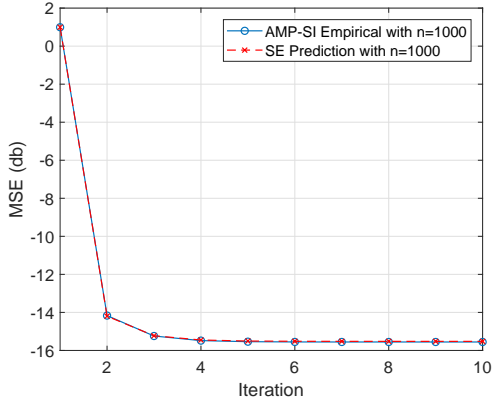}\label{GG_n=1000}}%
\hfill 
\subcaptionbox{n=10000}{\includegraphics[width=0.4\textwidth]{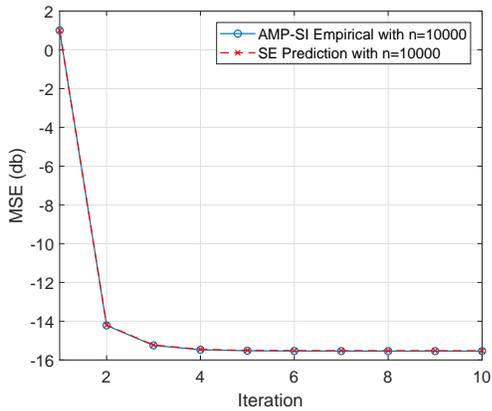}\label{GG_n=10000}}%
\caption{Empirical MSE performance of AMP-SI and SE prediction. (GG model, $\delta=0.3$, $\sigma_{x}=1$, $\sigma_{w}=0.1$, and $\sigma=0.2$.)}
\label{fig:GG}
\end{figure}

\begin{figure}
 \centering
 \includegraphics[width=0.8\linewidth]{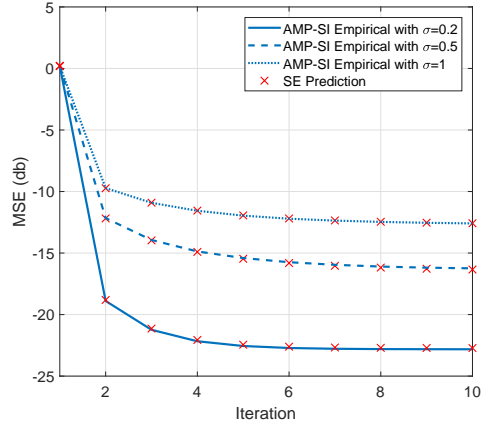}
 \caption{Empirical MSE performance of AMP-SI and SE prediction. (BG model, $n=10000$, $m=3000$, $\epsilon=0.2$, $\sigma_{w}=0.1$.)}
  \label{fig:BG}
\end{figure}
\section{Proof of Theorem~\ref{thm:SE}}\label{main_proof}
\subsection{Step 1}
First we show that the functions defined in \eqref{eq:sum_funcs} are uniformly PL(2) when $\phi$ and $\psi$ are PL(2).  This is a straightforward application of Cauchy-Schwarz.  We show the result for $\phi$ and the result for $\psi$ follows similarly.

First, by the fact that $\phi$ is PL(2) ,
\begin{equation*}
\begin{aligned}
&\lvert\phi_m(\mathbf{a}, \widetilde{\mathbf{a}}) - \phi_m(\mathbf{b}, \widetilde{\mathbf{b}}) \lvert 
\leq \frac{1}{m}\sum_{i=1}^m \lvert\phi(a_i, \widetilde{a}_i)- \phi(b_i, \widetilde{b}_i) \lvert \\
& \leq \frac{L}{m}\sum_{i=1}^m\Big[1 +  \frac{\lvert \lvert(a_i, \widetilde{a}_i) \lvert \lvert}{\sqrt{2}}   +  \frac{\lvert \lvert(b_i, \widetilde{b}_i) \lvert \lvert }{\sqrt{2}} \Big]
 \frac{\lvert \lvert(a_i, \widetilde{a}_i) - (b_i, \widetilde{b}_i) \lvert \lvert}{\sqrt{2}}.
\end{aligned}
\end{equation*}
Then applying Cauchy-Schwarz in the following way: for any $r >0$ and $a_1, a_2, \ldots, a_m$ scalars, $(|a_1| +|a_2| + \ldots |a_m|)^r \leq m^{r-1}(|a_1|^r +|a_2|^r + \ldots |a_m|^r)$, we have

\begin{equation*}
\begin{aligned}
&\lvert\phi_m(\mathbf{a}, \widetilde{\mathbf{a}}) - \phi_m(\mathbf{b}, \widetilde{\mathbf{b}}) \lvert^2 \\
& \leq  \frac{L^2}{2 m^2}\sum_{i=1}^m\Big[1 +  \frac{\lvert \lvert(a_i, \widetilde{a}_i) \lvert \lvert}{\sqrt{2}}   +  \frac{\lvert \lvert(b_i, \widetilde{b}_i) \lvert \lvert}{\sqrt{2}} \Big]^2\lvert\lvert(\mathbf{a}, \widetilde{\mathbf{a}}) - (\mathbf{b}, \widetilde{\mathbf{b}}) \lvert\lvert^2
\\& \leq 3L^2\Big[1 +  \frac{\lvert \lvert(\mathbf{a}, \widetilde{\mathbf{a}}) \lvert \lvert^{2}}{{2m}}   +  \frac{\lvert \lvert(\mathbf{b}, \widetilde{\mathbf{b}}) \lvert \lvert^{2} }{{2m}} \Big]
\frac{\lvert\lvert(\mathbf{a}, \widetilde{\mathbf{a}}) - (\mathbf{b}, \widetilde{\mathbf{b}}) \lvert\lvert^2}{2m}.
\end{aligned}
\end{equation*}
In the final inequality in the above we have used that
\begin{align*}
&\frac{1}{m} \sum_{i=1}^m\Big[1 +  \frac{\lvert \lvert(a_i, \widetilde{a}_i) \lvert \lvert}{\sqrt{2}}   +  \frac{\lvert \lvert(b_i, \widetilde{b}_i) \lvert \lvert}{\sqrt{2}} \Big]^2 \\
&\leq  3\sum_{i=1}^m\Big[\frac{1}{m} + \frac{\lvert \lvert(a_i, \widetilde{a}_i) \lvert \lvert^{2}}{2m}   +  \frac{\lvert \lvert(b_i, \widetilde{b}_i) \lvert \lvert^{2} }{2m} \Big]
\\
&=  3\Big( 1 + \sum_{i=1}^m \frac{a_i^2 + \widetilde{a}_i^2}{2m} + \sum_{i=1}^m \frac{b_i^2 + \widetilde{b}_i^2 }{2m}\Big).
\end{align*}
Finally, we note that this implies
\begin{equation*}
\begin{aligned}
&\lvert\phi_m(\mathbf{a}, \widetilde{\mathbf{a}}) - \phi_m(\mathbf{b}, \widetilde{\mathbf{b}}) \lvert \\
& \leq \sqrt{3}L\Big[ 1 +  \frac{||(\mathbf{a}, \widetilde{\mathbf{a}})||}{\sqrt{2m}}   + \frac{||(\mathbf{b}, \widetilde{\mathbf{b}})||}{\sqrt{2m}}  \Big] \frac{ \lvert \lvert(\mathbf{a}, \widetilde{\mathbf{a}}) - (\mathbf{b}, \widetilde{\mathbf{b}}) \lvert \lvert}{\sqrt{2m}}.
\end{aligned}
\end{equation*}
\subsection{Step 2}
Next we show the asymptotic results given in \eqref{eq:main}.
First we use Berthier \emph{et al}.~\cite[Theorem 14]{Berthier2017} and then we make an appeal to the strong law of large numbers (SLLN):  We remind the reader of the strong law:
\begin{definition}
\label{thm:LLN}
{\textbf{\emph{Strong Law of Large Numbers}}}~\cite{JR2006}: Let
$X_1, X_2,...$ be a sequence of i.i.d.\ random variables with finite mean
$\mu$. Then
\begin{equation}\label{eq:2-1}
\Pr\left(\lim_{n\to\infty}\frac{1}{n}(X_1+X_2+...+X_n)=\mu\right)=1,
\end{equation}
In words, the partial averages $\frac{1}{n}(X_1+X_2+...+X_n)$ converge almost surely
to $\mu < \infty$.
\end{definition}

We will make use of Berthier \emph{et al.}~\cite[Theorem 14]{Berthier2017}, restated here for convenience. To apply the result in Berthier \emph{et al.}~\cite[Theorem 14]{Berthier2017}, one needs to justify the following assumptions:
\begin{itemize}
\item[\textbf{(C1)}] The measurement matrix $\A$ has Gaussian entries with i.i.d.\  mean $0$ and variance $1/m$.
\item[\textbf{(C2)}] Define a sequence of denoisers $\widetilde{\eta}_{n}^t: \mathbb{R}^n \rightarrow \mathbb{R}^n$ to be those that apply the denoiser $\eta_t$ defined in \eqref{eq:eta_2_iid} elementwise as follows: $\widetilde{\eta}_{n}^t({\x}) := \eta_t({\x}, \widetilde{{\x}})$.
For each $t$, $\widetilde{\eta}_{n}^t(\cdot)$ are uniformly Lipschitz. A function is \emph{uniformly} Lipschitz in $n$ if the Lipschitz constant does not depend on $n$. 

\item[\textbf{(C3)}]  $||\x||_2^2/n$ converges to a constant as $n\to\infty$.
\item[\textbf{(C4)}] The limit $\sigma_w=\lim_{m\to\infty}{||\w||_2}/{\sqrt{m}}$ is finite.
\item[\textbf{(C5)}]
For any iterations $s, t\in \mathbb{N} $ and for any $2 \times 2$ covariance matrix $\boldsymbol \Sigma$, the following limits exist.
\begin{align*}
&\lim_{n\to\infty}\frac{1}{n} \mathbb{E}_{\Z}[\x^T \widetilde{\eta}_{n}^t(\x+\Z)]< \infty,\\
&\lim_{n\to\infty}\frac{1}{n} \mathbb{E}_{\Z,\Z'}\left[\widetilde{\eta}_{n}^t(\x+ \Z)^T
\widetilde{\eta}_{n}^s(\x+ \Z')\right] < \infty,
\end{align*}
where $(\mathbf{Z}, \mathbf{Z}')\sim N(0, \boldsymbol \Sigma \otimes \mathbb I_n)$, with $\otimes$ denoting the tensor product and $\mathbb I_n$ the identity matrix.
\end{itemize}
\begin{theorem}
\label{thm:Berthier}
Under the assumptions $\textbf{(C1)} - \textbf{(C5)}$, for any sequences of uniformly pseudo-Lipschitz functions 
$\rho_m: \mathbb{R}^{m} \times \mathbb{R}^{m} \rightarrow \mathbb{R}$ and $\gamma_n: \mathbb{R}^{n} \times \mathbb{R}^{n} \rightarrow \mathbb{R}$,
\begin{equation*}
\begin{aligned}
&\lim_m ( \rho_m({\ramp}^t, {\w}) -  \mathbb{E}_{{\Z_1}}[\rho_m({\w} + \sqrt{\lambda_t^2 - \sigma_w^2} \, {\Z_1}, {\w})])\overset{p}{=}0, \\
&\lim_n\left( \gamma_n({\x}^{t} + {\A}^T {\ramp}^t, {\x}) - \mathbb{E}_{{\Z_2}}\left[\gamma_n({\x} + \lambda_t {\Z_2}, {\x})\right] \right) \overset{p}{=}0,
\end{aligned}
\end{equation*}
where $\Z_1\sim \mathcal{N}(0,\mathbb{I}_m)$, $\Z_2\sim \mathcal{N}(0,\mathbb{I}_n)$, $\x^t$ and $\ramp^t$ are defined in the AMP-SI recursion~\eqref{eq:1-5_iid}-\eqref{eq:1-6_iid}, and $\lambda_t$ in the SE~\eqref{eq:SE2}.
\end{theorem}

Now we demonstrate that our assumptions $\textbf{(A1)} - \textbf{(A4)}$ stated in Section~\ref{main_result} are enough to satisfy the assumptions $\textbf{(C1)} - \textbf{(C5)}$ needed to apply Theorem~\ref{thm:Berthier}. 

Assumptions $\textbf{(A1)}$ and $\textbf{(C1)}$ are identical.  We will show that $\textbf{(C2)}$ follows from $\textbf{(A4)}$, $\textbf{(C4)}$ follows from $\textbf{(A2)}$, and $\textbf{(C3)}$ follows from $\textbf{(A3)}$.  Finally we show $\textbf{(C5)}$ follow from $\textbf{(A3)}$ and $\textbf{(A4)}$.

First consider assumption $\textbf{(C2)}$. The non-separable denoiser $\widetilde{\eta}_{n}^t(X) = \eta_t({X}, \widetilde{X})$ applies the AMP-SI denoiser defined in \eqref{eq:eta_2} entrywise to its vector inputs. From $\textbf{(A4)}$, $\{\eta_t(\cdot, \cdot)\}_{t \geq 0}$ are Lipschitz continuous. Thus, for length-$n$ vectors $x_1, x_2$, and fixed SI $\widetilde{x}$,
\begin{equation*}
\begin{split}
&||\widetilde{\eta}_{n}^t({x_1}) - \widetilde{\eta}_{n}^t({x_2})||^2 = \sum_{i=1}^n (\eta_t([x_1]_i, \widetilde{x}_i) - \eta_t([x_2]_i, \widetilde{x}_i))^2 \\
& \leq  \sum_{i=1}^n L^2 ([x_1]_i - [x_2]_i)^2 =  L^2 ||x_1 - x_2||^2,
\end{split}
\end{equation*}
and so
\[||\widetilde{\eta}_{n}^t({x_1}) - \widetilde{\eta}_{n}^t({x_2})|| \leq  L ||x_1 - x_2||.\]
The Lipschitz constant does not depend on $n$, so $\widetilde{\eta}_{n}^t({\cdot})$ is uniformly Lipschitz.

Now consider assumption $\textbf{(C4)}$. From $\textbf{(A2)}$, the measurement noise $\w$ in~\eqref{eq:1-1} has i.i.d.\  entries with zero-mean and finite $\mathbb{E}[|W|^2]$. Then applying Definition~\ref{thm:LLN}, 
\begin{equation*}
\lim_{m\to\infty}\frac{||\w||_2^2}{m}=\lim_{m\to\infty}\frac{1}{m}\sum_{i=1}^mw_i^2
=\sigma_w^2 < \infty,
\end{equation*}
where we have used that $\sigma_w^2 < \infty$ follows from $\mathbb{E}[|W|^2] < \infty$.
The proof of \textbf{(C3)} similarly follows using  the SLLN and the finiteness of $\mathbb E[|X|^2]$ given in assumption \textbf{(A3)}.

We now show that $\textbf{(C5)}$ is met. Recall $Z\sim \mathcal{N}( 0,\sigma_z^2\mathbb I_n)$. Define $y_i :=x_i\mathbb{E}_{Z} \left[ \eta_t(x_i + Z_{i}, \widetilde{x}_i)\right]$ for $i = 1, 2, \ldots, n$. By assumption $\textbf{(A3)}$, the signal and side information $(\X, \widetilde{\X})$ are sampled i.i.d.\ from the joint density $f(X, \widetilde{X})$. It follows that $y_1, y_2, \ldots, y_n$ are also i.i.d., so by Definition~\ref{thm:LLN} if $\mathbb{E}[X\eta_t(X + Z,\widetilde X)] < \infty$ where $Z \sim \mathcal{N}(0,\sigma_z^2)$  independent of $(X, \widetilde{X}) \sim f(X, \widetilde{X})$, then
\begin{equation*}
\begin{aligned}
\lim_{n\to\infty}\frac{1}{n}\sum_{i=1}^n x_i\mathbb{E}_{Z} \left[\eta_t(x_i+Z_{i},\widetilde x_i)\right] =  \mathbb{E}[X\eta_t(X + Z,\widetilde X)].
\end{aligned}
\end{equation*}
We now show that $\mathbb{E}[X\eta_t(X + Z,\widetilde X)] < \infty$.

First note that $\textbf{(A4)}$ assumes $\eta_t(\cdot, \cdot)$ is Lipschitz, meaning for scalars $a_1, a_2, b_1, b_2$ and some constant $L > 0$,
\begin{equation*}
\begin{aligned}
|\eta_t(a_1, b_1) - \eta_t(a_2, b_2)| \leq L ||(a_1, b_1) - (a_2, b_2)|| \\\leq L |a_1 - a_2| + L |b_1 - b_2|.
\end{aligned}
\end{equation*}
Therefore letting $a_2 = b_2 = 0$ we have
\begin{equation*}
|\eta_t(a_1, b_1)| - |\eta_t(0, 0)| \leq |\eta_t(a_1, b_1) - \eta_t(0, 0)| \leq L |a_1| + L |b_1|,
\end{equation*}
giving the follows  upper bound for constant $L' > 0$,
\begin{equation}
  |\eta_t(a_1, b_1)|  \leq L'(1 + |a_1| + |b_1|).
\label{eq:eta_bound}
\end{equation}
Now using \eqref{eq:eta_bound} and the triangle inequality,
\begin{equation}
\begin{aligned}
& \mathbb{E} [X \eta_t(X + Z, \widetilde{X})] \leq L'\mathbb{E} [|X|  (1+ |X+Z| + |\widetilde{X}|)] \\
 &\leq L'(  \mathbb{E} |X|  + \mathbb{E} [X^2] + \mathbb{E} |X|  \mathbb{E}|Z| + \mathbb{E} |X \widetilde{X}|) .
\label{eq:C5_bound2}
\end{aligned}
\end{equation}
Finally, by assumption $\textbf{(A3)}$ we have that $\mathbb{E}[|X|^2], \mathbb{E}[|\widetilde{X}|^2]$ and $\mathbb{E}|X \widetilde{X}|$ are all finite.  Then noting that for any random variable, $Y$, we have  $|Y|^r \leq 1 + |Y|^k$ for $1 \leq r \leq k$, meaning $\mathbb{E}[|Y|]^r < 1+ \mathbb{E}[|Y|^k]$ the boundednes of $ \mathbb{E} [X \eta_t(X + Z, \widetilde{X})]$ follows from \eqref{eq:C5_bound2} with assumption $\textbf{(A3)}$.

The proof of the second equation in $\textbf{(C5)}$ follows similarly to the proof of the first equation in $\textbf{(C5)}$.  Recall $(Z, Z')\sim N(0,\Sigma \otimes \mathbb I_n)$. Define $y_i :=\mathbb{E}_{Z,Z'}[\eta_t(x_i+ Z_i,\widetilde{x}_i)\eta_s(x_i+ Z'_i,\widetilde{x}_i)] $ for $i = 1, 2, \ldots, n$.  By assumption $\textbf{(A3)}$, the signal and side information $(\X, \widetilde{\X})$ are sampled i.i.d.\ from the joint density $f(X, \widetilde{X})$. It follows that $y_1, y_2, \ldots, y_n$ are also i.i.d., so by Definition~\ref{thm:LLN} if $\mathbb{E} [\eta_t(X+Z,\widetilde X)\eta_s(X+Z',\widetilde X)] < \infty$ where $Z \sim \mathcal{N}(0,\sigma_z^2)$ and $Z' \sim \mathcal{N}(0,\sigma_{z'}^2)$, independent of $(X, \widetilde{X}) \sim f(X, \widetilde{X})$, then
\begin{equation*}
\begin{aligned}
\lim_{n\to\infty}\frac{1}{n} \sum_{i=1}^n \mathbb{E}_{Z,Z'}[\eta_t(x_i+ Z_i,\widetilde{x}_i)\eta_s(x_i+ Z'_i,\widetilde{x}_i)]\\ = \mathbb{E}[\eta_t(X+Z,\widetilde X)\eta_s(X+Z',\widetilde X)].
\end{aligned}
\end{equation*}
We will now show that $\mathbb{E}[\eta_t(X+Z,\widetilde X)\eta_s(X+Z',\widetilde X)] < \infty$.
Using the bound \eqref{eq:eta_bound},
\begin{equation*}
\begin{aligned}
&\mathbb{E}[\eta_t(X+Z,\widetilde X)\eta_s(X+Z',\widetilde X)] \\& \leq \mathbb{E}[|\eta_t(X+Z,\widetilde X)| |\eta_s(X+Z',\widetilde X)|] \\&
\leq {L'}^2\mathbb{E} [(1+ |X+Z| + |\widetilde{X}|)  (1+ |X+Z'| + |\widetilde{X}|)].
\label{eq:C6_bound1}
\end{aligned}
\end{equation*}
Then using the triangle inequality,
\begin{equation}
\begin{aligned}
&\mathbb{E} \left[(1+ |X+Z| + |\widetilde{X}|)  (1+ |X+Z'| + |\widetilde{X}|)\right] \\&\leq \mathbb{E} \left[(1+ |X| + |Z| + |\widetilde{X}|)  (1+ |X| + |Z'| + |\widetilde{X}|)\right] \\
&=   1+  2\mathbb{E} \left[|X|\right]  + 2 \mathbb{E} [|\widetilde{X}|]  + 2\mathbb{E} \left[|X| |\widetilde{X}|\right]  + \mathbb{E} [ X^2] + \mathbb{E} [\widetilde{X}^2] \\
& \qquad+ \mathbb{E} [|X||Z'|]  + \mathbb{E} [|X||Z|] +\mathbb{E} [ |\widetilde{X}||Z'| ] + \mathbb{E} [ |\widetilde{X}| |Z|  ] \\&  \qquad+  \mathbb{E} |Z|   + \mathbb{E} [|Z'|]   + \mathbb{E} \left[|Z| |Z'|\right]  \\
&=  1+  2\mathbb{E} [|X|]  + 2 \mathbb{E} [|\widetilde{X}|]  + 2\mathbb{E} \left[|X \widetilde{X}|\right]  + \mathbb{E} [ X^2] + \mathbb{E} [\widetilde{X}^2] \\
&  \qquad+  \mathbb{E} [|Z'|](1 + \mathbb{E} [|X|]+ \mathbb{E}  [|\widetilde{X}|] )  +  \mathbb{E} [|Z|] ( 1 \\& \qquad+ \mathbb{E} [|X|] + \mathbb{E}  [|\widetilde{X}|] )    + \mathbb{E} \left[|ZZ'|\right] .
\label{eq:C6_bound2} 
\end{aligned}
\end{equation}

\subsection{Step 3}
Now that we've justified $\textbf{(C1)} - \textbf{(C5)}$, we make an appeal to Theorem~\ref{thm:Berthier} and the SLLN in order to finally prove \eqref{eq:main}.
The first result in \eqref{eq:main}, namely the asymptotic result for $\phi_m$ uniformly PL(2), follows \emph{almost} immediately by applying Theorem~\ref{thm:Berthier} using $ \rho_m = \phi_m$. Namely, by Theorem~\ref{thm:Berthier}, 
$$\lim_m ( \phi_m({\ramp}^t, {\w}) -  \mathbb{E}_{{\Z_1}}[\phi_m({\w} + \sqrt{\lambda_t^2 - \sigma_w^2} \, {\Z_1}, {\w})])\overset{p}{=}0$$
since $\phi_m$ is assumed to be uniformly PL(2).  To complete the proof, we will finally prove that
\begin{equation}
\begin{split}
&\lim_m  \mathbb{E}_{{\Z_1}}[\phi_m({\w} + \sqrt{\lambda_t^2 - \sigma_w^2} \, {\Z_1}, {\w})] \\
&\qquad = \lim_m  \frac{1}{m} \sum_{i=1}^m\mathbb{E}_{{\Z_1}}[\phi({w}_i + \sqrt{\lambda_t^2 - \sigma_w^2} \, {[\Z_1]_i}, {w_i})] \\
&\qquad \overset{a.s.}{=} \mathbb{E}[\phi(W + \sqrt{\lambda_t^2 - \sigma_w^2} \,  {Z_1}, W)],
\label{eq:limit_num1}
 \end{split}
 \end{equation}
where $W \sim f(W)$ independent of $Z_1$ standard Gaussian.
Then the desired result follows since
$$ \mathbb{E}[\phi_m({\mathbf{W}} + \sqrt{\lambda_t^2 - \sigma_w^2} \, {\Z_1}, {\mathbf{W}})] = \mathbb{E}[\phi(W + \sqrt{\lambda_t^2 - \sigma_w^2} \,  {Z_1}, W)].$$
The result follows by the SLLN (Definition~\ref{thm:LLN}) so long as $\mathbb{E}[\phi(W + \sqrt{\lambda_t^2 - \sigma_w^2} \,  {Z_1}, W)]$ is finite.
By Definition \ref{def:PLfunc} it is easy to see that if $\phi: \mathbb{R}^2 \rightarrow \mathbb{R}$ is PL(2), then there is a constant $L' > 0$ such that for all $\mathbf{x}\in\mathbb{R}^2$ : $|\phi(\x)|\leq L'(1+||\x||^2).$ Using this,
\begin{equation}
\begin{aligned}
\label{eq:W_equation}
&\lvert \phi(W + \sqrt{\lambda_t^2 - \sigma_w^2} \,  {Z_1}, W)\lvert 
\leq  L_1'(1+ \lvert  \lvert (W + \sqrt{\lambda_t^2 - \sigma_w^2} \,  {Z_1}, W)   \lvert  \lvert^2)
\\ &\leq  L_1'(1+3 |W|^2 + 2 (\lambda_t^2 - \sigma_w^2)  |Z_1|^2),
\end{aligned}
\end{equation}
where we have used: for any $r >0$ and any $a_1, a_2$ scalars, $ \lvert \lvert (a_1, a_2)   \lvert  \lvert^2 =  a_1^2 + a_2^2$ and $(|a_1| +|a_2|)^r \leq 2^{r-1}(|a_1|^r +|a_2|^r)$. Thus,
\begin{align*} 
\Big\lvert \Big\lvert \Big(W + \sqrt{\lambda_t^2 - \sigma_w^2} \,  {Z_1}, W\Big) \Big  \lvert \Big \lvert^2 &= \Big(W + \sqrt{\lambda_t^2 - \sigma_w^2} \, {Z_1}\Big)^2 + W^2 \\
&\leq 3W^2 + 2(\lambda_t^2 - \sigma_w^2)  Z_1^2.
\end{align*}
Similarly, we have the upper bound,
\begin{equation}
\begin{aligned}
\Big |\psi(X + &\lambda_t {Z_2}, X, \widetilde{X})\Big| \leq  L_2'\left(1+\Big \lvert \Big \lvert(X + \lambda_t {Z_2}, X, \widetilde{X}) \Big \lvert \Big \lvert^2\right)\\& \leq L_2'\Big(1+3|X|^2 +  2\lambda_t^2 {|Z_2|}^2  +  |\widetilde{X}|^2 \Big).
 \label{eq:X_equation}
\end{aligned}
\end{equation}
Therefore, using \eqref{eq:X_equation}, and the boundedness of $\mathbb{E}[|X|^2]$ and $\mathbb{E}[|\widetilde{X}|^2]$ assumed in $\textbf{(A3)}$,
\begin{align*} 
&\mathbb{E}[\psi(X + \lambda_t {Z_2}, X, \widetilde{X})] \\&\leq L_2'\Big(1+3\mathbb{E}[|X|^2] +  2\lambda^2_t \mathbb{E}[{|Z_2|}^2]  +  \mathbb{E}[|\widetilde{X}|^2] \Big) < \infty.
\end{align*}

The second result of \eqref{eq:main} requires a bit more care as it is not immediate that the function $\gamma_n: \mathbb{R}^{2n} \rightarrow \mathbb{R}$ defined as $\gamma_n(\mathbf{a}, \mathbf{b}) := \psi_n(\mathbf{a}, \mathbf{b}, \widetilde{\x})
\label{eq:gamma_def}$ for a sequence of side informations $\{\widetilde{\x}\}_n$ is uniformly PL(2) as needed to apply Theorem~\ref{thm:Berthier}. The next step of the proof
deals with carefully handling this issue. We note that once we have shown that
\begin{equation}
\label{eq:to_show1}
\lim_n ( \psi_n(\x^t + {\A}^T {\ramp}^t,\x, \widetilde{\x}) -  \mathbb{E}_{{\Z_2}}[\psi_n({\x} + \lambda_t {\Z_2}, {\x}, \widetilde{\x}])\overset{p}{=}0,
\end{equation}
then the last step showing that
\begin{align*}
&\lim_n  \mathbb{E}_{{\Z_2}}[\psi_n({\x} + \lambda_t {\Z_2}, {\x}, \widetilde{\x})] \overset{p}{=} \lim_n  \mathbb{E}[\psi_n({\mathbf{X}} + \lambda_t {\Z_2}, {\mathbf{X}},\widetilde{\mathbf{X}})],
 \end{align*}
follows by the SLLN as in \eqref{eq:limit_num1} - \eqref{eq:W_equation}.  However, the function $\gamma_n$ is not obviously uniformly PL(2)  since an upper bound on $|\psi_n(\mathbf{a}, \widetilde{\mathbf{a}}, \widetilde{\x}) - \psi_n(\mathbf{b}, \widetilde{\mathbf{b}}, \widetilde{\x})|$ necessarily has an $||\widetilde{\x}||/\sqrt{n}$ factor.  This is mainly a technicality as $||\widetilde{\x}||/\sqrt{n}$ is bounded by a constant (independent of $n$) with high probability. 

To show \eqref{eq:to_show1} we would like to show that for any $\epsilon > 0$, 
\begin{equation}
P(\lvert  \psi_n(\x^t + {\A}^T {\ramp}^t,\x, \widetilde{\x}) -  \mathbb{E}_{{\Z_2}}[\psi_n({\x} + \lambda_t {\Z_2}, {\x}, \widetilde{\x})]  \lvert > \epsilon) \rightarrow 0
\label{eq:conv_in_prob}
\end{equation}
as $n \rightarrow \infty$.  Define a pair of events $\mathcal{T}_n(\epsilon)$ and $\mathcal{B}_n(C)$ as 
\[\mathcal{T}_n(\epsilon) := \{\lvert  \psi_n(\x^t + {\A}^T {\ramp}^t,\x, \widetilde{\x}) -  \mathbb{E}_{{\Z_2}}[\psi_n({\x} + \lambda_t {\Z_2}, {\x}, \widetilde{\x})]  \lvert > \epsilon \}\]
and for constant $C>0$ independent of $n$,
$\mathcal{B}_n(C) := \{\widetilde{\x} \in \mathbb{R}^n : ||\widetilde{\x}||/\sqrt{n} < C\}.$
Then demonstrating \eqref{eq:conv_in_prob} means showing, for any $\epsilon > 0
$, that $\lim_n P(\mathcal{T}_n(\epsilon)) = 0$.  Note that,
\begin{equation*}
\begin{split}
&P(\mathcal{T}_n(\epsilon)) =  P(\mathcal{T}_n(\epsilon) \text{ and } \mathcal{B}_n(C)) + P(\mathcal{T}_n(\epsilon)  \text{ and not } \mathcal{B}_n(C) ) \\
&\leq  P(\mathcal{T}_n(\epsilon) \lvert \mathcal{B}_n(C))+ P(\text{not } \mathcal{B}_n(C)).
\end{split}
\end{equation*}
Considering the above, the first term approaches $0$ as $n$ gets large due to Theorem~\ref{thm:Berthier}, since one can argue
$P(\mathcal{T}_n(\epsilon) \lvert \mathcal{B}_n(C)) = P(\mathcal{T}_n(\epsilon) \lvert \mathcal{B}_p(C) \text{ for all }  p > p_0 )$ and conditional on the event $\mathcal{B}_p(C)$ being true for all integers $p > p_0$ (constant $p_0>0$), 
the function $\gamma_n$ defined in \eqref{eq:gamma_def} is uniformly PL(2) in $n$.  This uses that $\widetilde{\x}(n)$ is independent of the other random elements, namely $\A(n)$ and $\w(n)$.  Next, by choosing $C$ large enough, the second probability $P(\text{not } \mathcal{B}_n(C))$ goes to zero almost surely by the SLLN as $ ||\widetilde{\x}||/\sqrt{n}$ concentrates to the elementwise expectation of $\widetilde{\x}$.

\section*{Acknowledgment}
We thank You (Joe) Zhou for insightful conversations and valuable advice. Liu and Baron acknowledge support from NSF EECS $\#1611112$ and Rush from NSF $\#1217023$.

\bibliographystyle{IEEEtran}
\bibliography{bibliography}

\end{document}